\documentclass[12pt, 
]{article}
\usepackage{booktabs}

\usepackage[utf8]{inputenc}
\usepackage{amssymb}
\usepackage{amsmath}

\usepackage{geometry}
\newtheorem{theorem}{Theorem}

\newtheorem{fact}{Fact}
\newtheorem{corollary}{Corollary}
\newtheorem{definition}{Definition}
\newtheorem{example}{Example}
\newtheorem{lemma}{Lemma}
\newtheorem{remark}{Remark}

\newenvironment{proof}[1][Proof]{\emph{#1.} }{\  \hfill $\square $ \vspace{5 pt}}

\usepackage{natbib}
\usepackage{amssymb}
\usepackage[dvips]{graphics}
\usepackage{amsmath}

\usepackage{mathpazo}
\usepackage{enumerate}

\usepackage{amsfonts}

\usepackage{amssymb}

\usepackage{enumerate}

\usepackage{mathrsfs}

\usepackage[usenames]{color}
\usepackage{cancel}

\usepackage{pgf,tikz}

\usetikzlibrary{decorations.pathreplacing,angles,quotes,intersections}
\usetikzlibrary{arrows.meta}
\usetikzlibrary{arrows, decorations.markings}
\tikzset{myptr/.style={decoration={markings,mark=at position 1 with %
       {\arrow[scale=2,>=stealth]{>}}},postaction={decorate}}}

\usepackage{hyperref}
\hypersetup{
    colorlinks,
    citecolor=blue,
    filecolor=black,
    linkcolor=blue,
    urlcolor=blue
}

\usepackage{pgf,tikz}
\usetikzlibrary{arrows}

\usepackage{fancyhdr}

\usepackage{soul}

\usepackage{emptypage}

\usepackage[T1]{fontenc}
\DeclareFontFamily{T1}{calligra}{}
\DeclareFontShape{T1}{calligra}{m}{n}{<->s*[1.44]callig15}{}
\DeclareMathAlphabet\mathcalligra   {T1}{calligra} {m} {n}

\usepackage{multirow}
\usepackage{array}
\usepackage{mathtools}
\usepackage{arydshln}
\usepackage{threeparttable}
\usepackage{booktabs,caption}

\usepackage{ifthen}
\newboolean{showcomments}
\setboolean{showcomments}{true}

\newcommand{\pablo}[1]{  \ifthenelse{\boolean{showcomments}}
{\textcolor{green!50!black}{(T: #1)}}{}}
\newcommand{\marcelo}[1]{\ifthenelse{\boolean{showcomments}}
{\textcolor{red}{(M: #1)}}{}}
\newcommand{\agustin}[1]{  \ifthenelse{\boolean{showcomments}}
{\textcolor{blue!50!black}{(T: #1)}}{}}

\usepackage{orcidlink}  

\begin{document}

\title{Pareto-undominated 
strategy-proof rules in economies with multidimensional single-peaked preferences\footnote{A previous version of this paper was circulated under the name ``Pareto dominant strategy-proof rules with multiple commmodities''. Part of this work was written during my visit at the Department of Economics, University of Rochester, under a Fulbright scholarship. I thank William Thomson for his hospitality. I also thank Federico Fioravanti, Jordi Massó, and Alejandro Neme for their detailed comments.  
I acknowledge financial support from Universidad Nacional de San Luis  through grants 032016 and 030320, from Consejo Nacional
de Investigaciones Cient\'{\i}ficas y T\'{e}cnicas (CONICET) through grant
PIP 112-200801-00655, and from Agencia Nacional de Promoción Cient\'ifica y Tecnológica through grant PICT 2017-2355.}}

\author{Agustín G. Bonifacio\footnote{Instituto de Matemática Aplicada San Luis (UNSL-CONICET), Av. Italia 1556,  5700, San Luis, Argentina. E-mail: \href{mailto:agustinbonifacio@gmail.com}{ \texttt{abonifacio@unsl.edu.ar}}} \hspace{1.5 pt} \orcidlink{0000-0003-2239-8673}}

\date{\today}

\maketitle

\begin{abstract}

In the problem of fully allocating a social endowment of perfectly divisible commodities among a group of agents with multidimensional single-peaked preferences, we study strategy-proof rules that are not Pareto-dominated by other strategy-proof rules. Specifically, we: (i) establish a sufficient condition for a rule to be \emph{Pareto-undomina\-ted strategy-proof};  (ii) introduce a broad class of rules satisfying this property by extending the family of ``sequential allotment rules'' to the multidimensional setting; and (iii) provide a new characterization of the ``multidimensional uniform rule'' involving \emph{Pareto-undominated strategy-proofness}. Results (i) and (iii) generalize previous findings that were only applicable to the two-agent case.

\vspace{5 pt}

\noindent
\emph{Journal of Economic Literature}  Classification Numbers: D63, D71, D78 

\end{abstract}

\vspace{10 pt}

\section{Introduction}

Consider the problem of fully allocating a social endowment of perfectly divisible commodities among a group of agents. An allocation rule specifies an allocation for each such problem. Since the pioneering work of \cite{Hurwicz72}, research on allocation rules in economic environments has highlighted a fundamental incompatibility: the standard notion of efficiency conflicts with the strong incentive-compatibility property of ``strategy-proofness''---by which no agent can benefit from misrepresenting her preferences---when fairness concerns are also taken into account.

For instance, in two-agent economies, any strategy-proof and efficient rule must be ``dictatorial'' (i.e., one agent always receives her most preferred allotment) under the following preference conditions: (i) homothetic or linear preferences \citep{Schummer97}, (ii) preferences with constant elasticity of substitution \citep{Ju03}, or (iii) quasi-linear preferences \citep{Goswami14}.\footnote{An important exception is the domain of preferences representable by Leontief-type utility functions, in which there are strategy-proof and efficient rules that are not dictatorial \citep[see][]{Nicolo04,Li13}.} Moreover, \cite{cho2023strategy} demonstrate that, in economies with more than two agents and linear preferences, no rule can simultaneously satisfy strategy-proofness, efficiency, and ``equal treatment'' (i.e., agents with identical preferences must receive indifferent allotments).

In this paper, we address this incompatibility in the context of economies where preferences are multidimensionally single-peaked. These preferences satisfy the following condition: for each commodity, holding the consumption levels of other commodities fixed, an agent’s welfare increases up to a critical consumption level, beyond which further increases reduce her welfare. The multidimensional single-peaked domain is important for two reasons: (i) under regularity conditions (such as continuity and strict convexity), it generalizes the classical domain of homothetic preferences when the consumption set is compact, and (ii) it accommodates various interpretations across different applications.\footnote{For one application, consider an exchange economy with classical preferences but assume that prices, somehow, are  kept fixed and a  rationing scheme is needed. If the preferences of the agents are strictly convex, their restriction to the boundary of their budget set will be multidimensional single-peaked.  Another application consists of imagining a group of partially altruist agents who have to divide among them a  bundle of goods. Suppose agents pay more attention to themselves when they consume little and give more priority to the consumption of the rest as their own consumption increases. In that case, their  preferences 
are also multidimensional single-peaked  \citep[see][Chapter 11]{Thomson13}.}

When there is only one commodity to be fully allocated, economies with single-peaked preferences allow for non-trivial, strategy-proof, and efficient rules. Among these, the ``uniform rule'' \citep{Sprumont91} stands out due to its desirable properties: it is the only rule that satisfies strategy-proofness, efficiency, and equal treatment \citep{Ching94}. Furthermore, \cite{barbera1997strategy} study the family of ``sequential rules'', which allow for asymmetric treatment of agents, and characterize them as the only strategy-proof and efficient rules that also satisfy ``replacement monotonicity''. This property states that if a change in an agent’s preference does not decrease (or increase) the amount assigned to that agent, the amounts assigned to the other agents cannot increase (or decrease).

However, in economies with multidimensional single-peaked preferences, this compatibility breaks down. \cite{amoros2002single} shows that in this setting, any rule satisfying both strategy-proofness and efficiency must be dictatorial, even with only two agents. Given that our primary focus is the study of strategy-proof rules, we adopt a weaker notion of efficiency.\footnote{Another approach to addressing the impossibility is to retain efficiency while weakening strategy-proofness. A recent trend in the literature examines rules that allow for some degree of manipulability, but not in an ``obvious'' way, following the insight of \cite{troyan2020obvious}. An application of these ideas to the model considered here, specifically in the one-commodity case, is presented in \cite{arribillaga2025not}.} To address this, we focus on rules that satisfy the informational simplicity condition of being ``own-peak-only''---meaning that the only information the rule uses from an agent's preference to determine her allocation is her peak consumption bundle. Within the class of strategy-proof and own-peak-only rules, we define a notion of Pareto domination and, following \cite{anno2013second}, identify those rules that are not strictly dominated by others. We refer to these as ``Pareto-undominated strategy-proof'' rules. One of the main results of \cite{anno2013second} states that, in two-agent economies, the multidimensional version of the uniform rule \citep[studied by][among others]{amoros2002single,Morimoto13} is the only Pareto-undominated strategy-proof rule that also satisfies equal treatment.

Our aim is to generalize this characterization of the multidimensional uniform rule to economies with any number of agents. To achieve this, we extend the notion of replacement monotonicity from \cite{barbera1997strategy} to multidimensional settings. A rule satisfies ``multidimensional replacement monotonicity'' if, when an agent’s preference changes, the resulting change in the allocation \emph{for each commodity} moves in the opposite direction for the other agents relative to the change in the agent’s own allocation. Our main result (Theorem \ref{caractvariosbienes}) characterizes the multidimensional uniform rule as the only rule that satisfies Pareto-undominated strategy-proofness, equal treatment, and multidimensional replacement monotonicity.

To derive this result, we first establish a sufficient condition for a rule to be (strongly) Pareto-undominated strategy-proof (Theorem \ref{main}). This condition involves combining the properties of  strategy-proofness, own-peak-onliness, and multidimensional replacement monotonicity with the mild property of ``unanimity''. Unanimity requires that if the sum of the agents’ peak amounts for each commodity equals the total supply, then each agent’s assignment must match her peak vector. As a by-product, we also show that the multidimensional version of the sequential rules introduced by \cite{barbera1997strategy} are strongly Pareto-undominated strategy-proof (Theorem \ref{theo sequential}).

The paper is organized as follows. In Section \ref{prelimvarios} we introduce the model, the properties we will rely on, and a general impossibility theorem. The property of Pareto-undominated strategy-proofness is introduced together with a sufficient condition to achieve it in Section \ref{EffNoMan}. The multidimensional sequential rules and the multidimensional uniform rule are analyzed in Sections \ref{SecuencialesVariosBienes} and \ref{UniformeVariosBienes}, respectively. Finally, some comments are gathered in Section  \ref{ComentUnifVarios}. Most proofs are relegated to Appendix \ref{appendix}.

\section{Model, properties and preliminary results}\label{prelimvarios}

Let  $L \equiv \{1, \dots, l\}$ be the set of goods.  For each  $\ell \in L,$ there is an amount  $\Omega^\ell \in \mathbb{R}_{+}$ to be allocated. Let  $\Omega \equiv (\Omega^1, \dots, \Omega^l) \in \mathbb{R}_{+}^{l}$ be the social endowment and $X \equiv \prod_{\ell \in  L} [0, \Omega^\ell]$ the consumption set. The set of agents is $N \equiv \{1, \ldots, n\}$. Each agent $i \in N$ has a preference relation $R_i$  defined on $X$ which is complete, transitive, continuous, and strictly convex. Denote the strict and indifference relations associated to $R_i$ by  $P_i$ and  $I_i,$ respectively. The most preferred consumption in $X$ according to $R_i$ (whenever it exists) is called the \emph{peak} of $R_i$ and is denoted by $p(R_i)$, and its $\ell$-component by $p^\ell(R_i)$, for each $\ell \in L$. In general, we denote a domain of preferences by $\mathcal{R}.$  

A profile of preferences is a list  $R \equiv (R_1, \ldots, R_n) \in \mathcal{R}^{n}$, and its related profile of peaks (whenever it exists) is denoted by $p(R)$. Given a profile $R \in \mathcal{R}^{n}$ and an agent $i \in N,$ the notation $R_{-i}$ refers to the list of preferences of all the agents except agent $i.$  Sometimes, we will use a similar notation replacing $i$ with a set of agents $S \subseteq N.$ 

An  \emph{economy} is a pair   $(R, \Omega) \in \mathcal{R}^{n} \times \mathbb{R}_{+}^{l}.$    A \emph{(feasible) allocation} is a list $x \equiv (x_1, \dots, x_n) \in X^{n}$ such that $\sum_{i \in N} x_i=\Omega.$ Let $Z$ denote the set of feasible allocations. An \emph{(allocation) rule} on $\mathcal{R}$, denoted by $\varphi,$ is a function from   $\mathcal{R}^{n}$ to  $Z.$ Throughout this paper we will kept fixed the social endowment  $\Omega \in \mathbb{R}^{l}_+.$ Therefore, an economy will simply be represented by a profile of preferences $R \in \mathcal{R}^{n}.$ 

Let $\widehat{x}, \overline{x}, \widetilde{x} \in X.$ We say that \emph{$\overline{x}$ is between $\widehat{x}$ and $\widetilde{x}$} if for each $\ell \in L,$ $\widehat{x}^\ell \leq  \overline{x}^\ell \leq \widetilde{x}^\ell$  or $\widehat{x}^\ell \geq  \overline{x}^\ell \geq \widetilde{x}^\ell$. To define the class of preferences we will study, fix an agent $i \in N.$

\begin{definition}\label{defunimul}
 A preference  $R_i$ is   \textbf{multidimensional single-peaked} if there is a peak vector $p(R_i) \in X$ such that, for each pair  $x_i, \widetilde{x}_i \in X$ with  $x_i \neq \widetilde{x}_i,$ whenever 
 $x_i$ is between $p(R_i)$ and $\widetilde{x}_i$,  we have $x_i \ P_i \ \widetilde{x}_i.$
\end{definition}
From now on, $\mathcal{R}$ stands for the domain of all multidimensional single-peaked preferences.

Next,  we present the properties of allocation rules that we will consider. The first limits the strategic behavior of the agents requiring that none of them can manipulate the rule by declaring false preferences:

\vspace{5 pt}
\noindent \textbf{Strategy-proofness:} For each  $R \in \mathcal{R}^{n},$ each $i \in N,$ and each  $R_i' \in \mathcal{R},$  we have 
$$\varphi_i(R) \ R_i \ \varphi_i(R_i',R_{-i}).$$

The second property requires that, for each economy, a rule always selects an allocation such that no other allocation Pareto dominates it.

\vspace{5 pt}
\noindent
\textbf{Efficiency:} For each  $R \in \mathcal{R}^n,$ there is no $x' \in Z$ such that  $x_i' \ R_i \ \varphi_i(R)$ for each $i \in N$ and  $x_j' \ P_j \  \varphi_j(R)$ for some  $j \in N$.
\vspace{5 pt}

The following fairness property establishes that two agents with the same preferences must receive indifferent allocations.

\vspace{5 pt}
\noindent
\textbf{Equal treatment:} For each $R \in \mathcal{R}^{n}$ and each $\{i,j\} \subseteq N,$  $R_i=R_j$ implies  $\varphi_i(R) \ I_i \  \varphi_j(R).$
\vspace{5 pt}

In the domain of classical, homothetic, and smooth preferences, \cite{Serizawa02} shows that the three previously defined properties are incompatible. Let $\mathcal{R}_{cl}$ denote the domain of classical, homothetic and smooth preferences.\footnote{A preference relation $R_i$ is \textbf{homothetic} if $x_i R_i y_i$ implies $\lambda x_i R_i \lambda y_i$ for each $\lambda \in \mathbb{R}_+,$ and is \textbf{smooth} if, for each $x_i \in X\setminus \partial X,$ there is a unique  vector in the simplex $\{q \in \mathbb{R}^{l}_+ : ||q||=1\}$ which supports the upper contour set $\{y_i \in X : y_i R_i x_i\}.$}

\begin{fact}\label{fact 1}{\normalfont{\citep{Serizawa02}}}\label{Seri} No rule defined on $\mathcal{R}_{cl}$ satisfies  \emph{strategy-proofness, efficiency} and \emph{equal treatment.}
\end{fact}

Our first result shows that the impossibility of Fact \ref{fact 1}  extends to the domain of all multidimensional single-peaked preferences.

\begin{lemma}\label{lemma impossibility} No rule defined on $\mathcal{R}$ satisfies  \emph{strategy-proofness, efficiency} and \emph{equal treatment.}
\end{lemma} 
\begin{proof}
This is a straightforward consequence of the fact that $\mathcal{R}_{cl} \subseteq \mathcal{R}.$ Let  $\varphi$ be a \emph{strategy-proof,} \emph{efficient} and  \emph{equally-treating} rule defined on  $\mathcal{R}.$  Let $\hat{\varphi}$  be the restriction of  $\varphi$ to  $\mathcal{R}_{cl}.$ Then,  $\hat{\varphi}$ is  \emph{strategy-proof,} \emph{efficient} and \emph{equally-treating} on $\mathcal{R}_{cl}.$ This contradicts Fact \ref{Seri}.
\end{proof}

As the combination of \emph{strategy-proofness} and \emph{efficiency} is incompatible with \emph{equal treatment,} and since we are interested in studying  \emph{strategy-proof} rules, we have to weaken \emph{efficiency.} The following property does that (and is equivalent to \emph{efficiency} in one-commodity economies): 

\vspace{5 pt}
\noindent
\textbf{Same-sidedness:}  For each $R \in \mathcal{R}^n$  and each $\ell \in L,$
\begin{enumerate}[(i)]
    \item $\sum_{i \in N} p^\ell(R_i) \geq \Omega^\ell$ implies   $\varphi_i^\ell(R) \leq p^\ell(R_{i})$ for each $ i \in N,$ and

    \item $\sum_{i \in N} p^\ell(R_{i}) \leq \Omega^\ell$ implies $\varphi_i^\ell(R) \geq p^\ell(R_{i})$ for each $ i \in N.$ 
\end{enumerate}
\vspace{5 pt}

\noindent Of course, every  \emph{efficient} rule is \emph{same-sided,} but the converse is false.\footnote{Define a rule in the following way: for each commodity, if the sum of the peaks is greater than zero, the rule assigns proportionally with respect to the peaks; otherwise the rule assigns the egalitarian allocation. This rule is \emph{same-sided.} Consider a two-good, two-agent economy such that the peak amount of both agents equals the social endowment and there is a feasible allocation that dominates the egalitarian allocation. In this economy, the rule selects the egalitarian allocation. However, such an allocation is not \emph{efficient}.}
 \emph{Same-sidedness} was used by  \cite{amoros2002single} to characterize the multidimensional version of the uniform rule together with \emph{strategy-proofness} and the property of  \emph{equal treatment}.
 
An even weaker property than \emph{same-sidedness} says that when, for each good, the sum of the peaks of the agents is equal to the social endowment, then the rule must assign, for each agent and each good, that peak amount.

\vspace{5 pt}
\noindent
\textbf{Unanimity:} For each  $ R \in \mathcal{R}^n,$ if  $\sum_{i \in N} p^\ell(R_{i}) = \Omega^\ell$ for each $\ell \in L,$ then  $\varphi_i^\ell(R)= p^\ell(R_i)$ for each $ i \in N$ and each  $\ell \in L.$
\vspace{5 pt}
\\
Clearly, every \emph{same-sided} rule is \emph{unanimous,} but the converse is false.\footnote{Think of a rule which assigns the peaks when feasible but is constant otherwise.} This property can be consider as a minimal requirement of efficiency.

Finally, we present two further properties that will be fundamental for the results of this paper. The first one establishes that, for each commodity, if the change in the preferences of one agent does not decrease (increase)  the amount assigned to that agent, then the amounts assigned to the rest of the agents cannot increase (decrease).

\vspace{5 pt}
\noindent
\textbf{Multidimensional replacement monotonicity:} For each $R \in \mathcal{R}^{n},$ each $i \in N,$ each $ R_{i}' \in  \mathcal{R},$ and each $\ell \in L,$
\begin{enumerate}[(i)]
    \item $\varphi_i^\ell(R)\leq \varphi_i^\ell(R_{i}',R_{-i})$ implies $\varphi_j^\ell(R)\geq \varphi_j^\ell(R_{i}',R_{-i})$ for each $j \in N\setminus \{i\},$ and

    \item $\varphi_i^\ell(R)\geq \varphi_i^\ell(R_{i}',R_{-i})$ implies $\varphi_j^\ell(R)\leq \varphi_j^\ell(R_{i}',R_{-i})$ for each $j \in N\setminus \{i\}.$ 
\end{enumerate}
\vspace{5 pt}

When there is only one commodity in the economy and \emph{efficiency} is also invoked, \emph{replacement monotonicity} can be interpreted as a solidarity principle: any (non-disruptive) change in an agent’s preference must affect all remaining agents’ welfare in the same direction. While this interpretation is compelling, its primary justification in models with multiple commodities is technical: it provides structure to otherwise complex classes of rules, such as strategy-proof rules, which are challenging to characterize.

The second property, introduced by \cite{Satterthwaite81}, has played an important role in the development of the axiomatic study of resource allocation. It establishes that, if the change in the preference of one agent does not change the amount assigned to that agent, then nobody else's assignment should change either.

\vspace{5 pt}
\noindent
\textbf{Non-bossiness:}  For each $R \in \mathcal{R}^{n},$ each  $i \in N,$  and each $ R_{i}' \in  \mathcal{R},$  $$\varphi_{i}(R)=\varphi_{i}(R_{i}',R_{-i}) \text{ \ implies \ }\varphi(R)=\varphi(R_{i}',R_{-i}).$$
\begin{remark}\label{remnonbossy}\emph{Every \emph{multidimensional replacement monotonic} rule is  \emph{non-bossy}. When there are only two agents, any rule satisfies trivially \emph{multidimensional replacement monotonicity} (and therefore \emph{non-bossiness}).}
\end{remark}
Several different interpretations have been given to  \emph{non-bossiness.} From a strategic perspective, this property strengthens in some models 
\emph{strategy-proofness} to  \emph{group strategy-proofness}.\footnote{\emph{Strategy-proofness} and \emph{non-bossiness} imply \emph{group strategy-proofness,} for example, in one-commodity economies with single-peaked preferences. However, when there are several commodities, this no longer holds. The multidimensional uniform rule, which we present in Section \ref{UniformeVariosBienes}, is not \emph{group strategy-proof} (see Example \ref{ej}).  For a thorough account on the use of  \emph{non-bossiness} in the literature, see \cite{thomson2016non}.} From a normative perspective, the property has the advantage of dismissing rules with certain ``arbitrary'' behavior.

\section{Pareto-undominated strategy-proof rules}\label{EffNoMan}

Here we present another way to weaken \emph{efficiency} within the class of \emph{strategy-proof} rules. The idea is to select those \emph{strategy-proof} rules which are undominated in the Pareto sense. Next, we present the formal definitions. Denote by $\Phi$ the class of all rules defined on $\mathcal{R}$. 

\begin{definition}
Let  $\varphi$ and $\psi$ be two rules in $\Phi$. We say that \textbf{$\boldsymbol{\varphi$ dominates $\psi}$}, and write  $\boldsymbol{\varphi \succcurlyeq  \psi}$, if and only if $\varphi_i(R) \ R_i \ \psi_i(R)$ for each $R \in \mathcal{R}^{n}$ and each $i \in N.$  
\end{definition}

\begin{remark}\emph{Relation  $\succcurlyeq$ is a preorder, i.e. it is reflexive and transitive, but not necessarily complete nor antisymmetric.\footnote{A binary relation  $\succeq$ defined on a set $X$ is \textbf{reflexive} if $x \succeq x$ for each $x \in X$; it is  \textbf{transitive} if $x \succeq y$ and $y \succeq z$ imply $x \succeq z$ for each  $\{x, y , z\} \subseteq X$;  it is  \textbf{complete} if $x \succeq y$ or $y \succeq x$ for each $\{x,y\} \subseteq X$; and it is \textbf{antisymmetric}  if $x \succeq y$ and $y \succeq x$ imply $x=y$ for each $\{x,y\} \subseteq X$.}}
\end{remark}
The property of \emph{efficiency} can be characterized in terms of the relation $\succcurlyeq$ as follows:

\begin{lemma}\label{lemma eff}  A rule  $\varphi \in \Phi$ is  \emph{efficient} if and only if $\varphi \succcurlyeq \psi$ for each  $\psi \in \Phi$ such that  $\psi \succcurlyeq \varphi$. 
\end{lemma}
\begin{proof} ($\Longrightarrow$) Let $\varphi \in \Phi$ be an \emph{efficient} rule and let $\psi \in \Phi$ be such that  $\psi \succcurlyeq \varphi$. Assume $\varphi \not\succcurlyeq \psi$. Then, there are   $R \in \mathcal{R}^{n}$  and $i \in N$ such that $\psi_i(R) \ P_i \  \varphi_i(R).$ Since $\psi_j(R) \ R_j \  \varphi_j(R)$  for each $j \in N\setminus\{i\},$ it follows $\varphi$ is not  \emph{efficient}, a contradiction. Hence, $\varphi \succcurlyeq \psi$.

\smallskip 

\noindent ($\Longleftarrow$)  Let $\varphi \in \Phi$. Assume that for each  $\psi \in \Phi$ such that  $\psi \succcurlyeq \varphi,$ we have $\varphi \succcurlyeq \psi$. If  $\varphi$ is not  \emph{efficient}, there are  $R' \in \mathcal{R}^{n}$ and  $x \in Z$ such that $x_j \ R_j' \  \varphi_j(R')$ for each  $j \in N$ and $x_i \  P_i' \  \varphi_i(R')$  for some $i \in N$. Define rule $\psi'$ as follows: for each $R \in \mathcal{R}^n,$
$$\psi'(R)\equiv\left\{ \begin{array}{l l}
 x & \text{ if } \  R=R' \\
\varphi(R) & \text{ if }  \ R\neq R'
\end{array} \right.$$
Then, $\psi' \succcurlyeq \varphi$ and $\varphi \not\succcurlyeq \psi'$, a contradiction. Hence, $\varphi$ is \emph{efficient}.
\end{proof}


Lemma \ref{lemma eff} shows that rule  $\varphi$ is \emph{efficient} whenever it is a ``maximal element'' of $\Phi$ preordered by $\succcurlyeq.$ Therefore, a way to weaken this property consists of  constraining the set over which a rule must be maximal. We will focus our analysis on the class of \emph{strategy-proof} rules that also fulfill the following informational simplicity property, that is satisfied for most of allocation rules studied in the literature.

\vspace{5 pt}
\noindent
\textbf{Own-peak-onliness:} For each  $ R \in \mathcal{R}^{n},$ each $i \in N,$ and each   $R_{i}' \in \mathcal{R},$  if $p(R_{i}')=p(R_{i}),$ then $\varphi_i(R) = \varphi_i(R_i',R_{-i}).$
\vspace{5 pt}
\\
Let  $\boldsymbol{\Phi^\star}$ denote the class of   \emph{strategy-proof} and  \emph{own-peak-only} rules. We now present a notion of constrained efficiency among \emph{strategy-proof} rules due to \cite{anno2013second}.\footnote{\cite{anno2013second} named this property  \emph{second-best efficiency among strategy-proof rules,} and do not require \emph{own-peak-onliness} of the rules in their definition.}

\begin{definition}
A rule $\varphi \in \Phi^{\star}$ is \textbf{Pareto-undominated strategy-proof} if  for each $\psi \in \Phi^{\star}$ such that $\psi \succcurlyeq \varphi$, we have $\varphi \succcurlyeq \psi$.
\end{definition}
Clearly, this property is implied by \emph{efficiency} when the rule is \emph{strategy-proof} and \emph{own-peak-only} on $\Phi$. A more demanding version of the property requires the rule to be a maximal element on $\Phi^\star$ with respect to preorder $\succcurlyeq$ and that no other rule be welfare-equivalent to it.\footnote{The rules  $\varphi, \psi \in \Phi$ are \textbf{welfare-equivalent} whenever $\varphi \succcurlyeq \psi$ and  $\psi \succcurlyeq \varphi.$}

\begin{definition} A rule $\varphi \in \Phi^\star$ is \textbf{strongly Pareto-undominated strategy-proof} if for each $\psi \in \Phi^\star$ such that $\psi \succcurlyeq \varphi$, we have $\varphi=\psi$.  
\end{definition}

Theorem 3 in \cite{anno2013second} establishes that, in two-agent economies (i.e. when $|N|=2$), every \emph{strategy-proof} and \emph{same-sided} rule is \emph{strongly Pareto-undominated strategy-proof.} The next theorem generalizes the aforementioned result to economies with an arbitrary amount of agents.

\begin{theorem}\label{main}
Let $\varphi$ be a \emph{strategy-proof, unanimous} and \emph{multidimensional replacement monotonic} rule. Then, $\varphi$ is a \emph{strongly Pareto-undominated strategy-proof} rule.
\end{theorem}
\begin{proof} See Appendix \ref{prueba teo 1}.
\end{proof}

\begin{remark}\emph{Strictly speaking, since the \emph{own-peak-only} property is embedded in our definition of \emph{Pareto-undominated strategy-proofness}, Theorem 3 in \cite{anno2013second} is \emph{not} implied by Theorem \ref{main}. However, within \emph{own-peak-only} rules, the generalization is valid since: (i) by Remark \ref{remnonbossy}, when $|N|=2,$ \emph{multidimensional replacement monotonicity} is trivially satisfied by any rule; and (ii)   \emph{unanimity} and  \emph{multidimensional replacement monotonicity} imply \emph{same-sidedness} (see Lemma \ref{SS} in the Appendix). 
}
\end{remark}

In the next section, we present a broad class of rules that are \emph{Pareto-undominated strategy-proof}.

\section{Multidimensional sequential rules}\label{SecuencialesVariosBienes}

In this section we study a multidimensional generalization of the sequential rules presented in \cite{barbera1997strategy} for one-commodity economies with single-peaked preferences. \cite{barbera1997strategy} define a one-dimensional sequential rule by taking a feasible allocation as the initial reference point and then applying a (weakly) welfare-improving iterative process that completes in at most $|N|$ steps. The main result in \cite{barbera1997strategy} says that a rule is sequential if and only if it is \emph{strategy-proof,} \emph{efficient} and \emph{replacement monotonic}  (recall that in one-commodity economies,  \emph{efficiency}  is equivalent to \emph{same-sidedness}).

Since, in one-commodity economies, \emph{strategy-proofness} and \emph{efficiency} imply the \emph{own-peak-only} property \citep[see Lemma 1 in][]{Ching94}, and \emph{replacement monotonicity} implies \emph{non-bossiness}, it follows that one-dimensional sequential rules are ``peaks-only'': the rule determines the overall allocation solely based on the profile of agents' peak consumption bundles.\footnote{\textbf{Peaks-onliness:} For each  $R, R' \in \mathcal{R}^{n}$ such that $p(R)=p(R'),$ we have $\varphi(R) = \varphi(R').$ \label{footnote peaks-only}} Therefore, for each $\ell \in L$, without loss of generality we can write a one-dimensional sequential rule as a function $\phi^\ell: (X^\ell)^n \longrightarrow Z^\ell$.

We will define multidimensional rules based on these one-dimensional rules as follows:

\vspace{5 pt}
\noindent
\textbf{Multidimensional sequential rule, $\boldsymbol{\phi}$:} For each $R \in \mathcal{R}$ and each $i \in N,$
 $$\phi_i(R) \equiv \left( \phi_i^1(p^1(R)), \ldots, \phi_i^l(p^l(R))\right),$$ where, for each $\ell \in L,$  $\phi^\ell:(X^\ell)^n \longrightarrow Z^\ell$ is a one-dimensional sequential rule à la \cite{barbera1997strategy} and $p^\ell(R)$ denotes the projection of $p(R)$ onto $[0,\Omega^\ell]^n$. 
\vspace{5 pt}

\begin{lemma}\label{secvarios}
Every multidimensional sequential rule is \emph{strategy-proof,} \emph{same-sided} and \emph{multidimensional replacement monotonic.}
\end{lemma}
\begin{proof} See Appendix \ref{prueba lema 3}.
\end{proof}

\noindent
The following result is an immediate consequence of Theorem \ref{main} and Lemma \ref{secvarios}.

\begin{theorem}\label{theo sequential}
Every multidimensional sequential rule is \emph{strongly Pareto-undominated} 
 \emph{strategy-proof.}
\end{theorem}

To present a nice feature of some sequential rules, we first introduce a property. It establishes that each agent must find her assignment at least as desirable as the egalitarian allocation.

\vspace{5 pt}
\noindent
\textbf{Egalitarian lower bound:} For each $ R \in \mathcal{R}^{N}$ and each $i \in N,$ $\varphi_i(R) \ R_i \ \frac{\Omega}{n}.$

\begin{remark}\label{remark egalitarian} \em
    It is easy to see that if a sequential rule has as initial reference point the egalitarian allocation $\frac{\Omega}{n}$, then the rule meets the \emph{egalitarian lower bound.}
\end{remark}

\section{The multidimensional uniform rule}\label{UniformeVariosBienes}

Introduced in the axiomatic literature by \cite{Sprumont91}, the uniform rule selects an efficient allocation that is as close to the egalitarian allocation as possible: each agent receives either her peak amount or a common amount chosen to ensure feasibility of the overall assignment.

\begin{figure}[ht]
\begin{center}
\begin{tikzpicture}

\fill[blue!5] (0,4) -- (6,4) -- (6,10) -- (0,10) -- (0,9.3) --  (-.2,9.2) -- (.2,9.1) -- (0,9) -- (0,4);


\draw[thick,-, ] 
(0,0) -- (9,0) ;

\draw[thick,-, ] 
(9,0)-- (9.1,.2) -- (9.2,-.2) -- (9.3,0) ;

\draw[thick,->, ] 
(9.3,0) -- (11,0) ;

\draw[thick,-,] (0,0) -- (0,9);

\draw[thick,-, ] 
(0,9)-- (.2,9.1) -- (-.2,9.2) -- (0,9.3) ;

\draw[thick,->, ] 
(0,9.3) -- (0,11) ;




\draw[thick, blue!20] (6,4) -- (6,10) -- (0,10);
\draw[thick, blue!20] (0,4) -- (6,4);
\draw[
dashed] (6,0) -- (6,4);

\begin{scope}
        \clip (0,0) rectangle (9,9);
\draw[blue, thick] (2,2) ellipse (2.5cm and 4cm);
\draw[blue, thick] (2,2) ellipse (1.25cm and 2cm);


\draw[red, thick] (4,7) ellipse (1.75cm and 1.093cm);
\draw[red, thick] (4,7) ellipse (3.5cm and 2.186cm);



\draw[orange, thick] (8,4) ellipse (1 cm and 1.75cm);
\draw[orange, thick] (8,4) ellipse (2cm and 3.5cm);

\end{scope}

\draw[
dashed] (0,2) -- (2,2);
\draw[
dashed] (6,4) -- (8,4);
\draw[
dashed] (4,0) -- (4,7);
\draw[
dashed] (0,7) -- (4,7);
\draw[
dashed] (2,0) -- (2,4);
\draw[
dashed] (8,0) -- (8,4);

\draw[fill=blue] (2,2) circle[radius=1.5 pt] node[right]{\textcolor{blue}{$p(R_1)$}};
\draw[fill=red] (4,7) circle[radius=1.5 pt] node[above]{\textcolor{red}{$p(R_2)=u_2(R)$}};
\draw[fill=orange] (8,4) circle[radius=1.5 pt] node[above]{\textcolor{orange}{$p(R_3)$}};

\draw[fill=black] (0,7) circle[radius=1.5 pt] node[left]{$7$};
\draw[fill=black] (0,2) circle[radius=1.5 pt] node[left]{$2$};
\draw[fill=black] (2,0) circle[radius=1.5 pt] node[below]{$2$};
\draw[fill=black] (4,0) circle[radius=1.5 pt] node[below]{$4$};
\draw[fill=black] (8,0) circle[radius=1.5 pt] node[below]{$8$};

\draw[fill=black] (0,4) circle[radius=1.5 pt] node[left]{$\lambda^2=4$};
\draw[fill=black] (6,0) circle[radius=1.5 pt] node[below]{$\lambda^1=6$};
\draw[fill=orange] (6,4) circle[radius=1.5 pt] node[above left]{\textcolor{orange}{$u_3(R)$}};
\draw[fill=blue] (2,4) circle[radius=1.5 pt] node[above]{\textcolor{blue}{$u_1(R)$}};

\draw[fill=black] (10,0) circle[radius=1.5 pt] node[below]{$\Omega^1=12$};

\draw[fill=black] (0,10) circle[radius=1.5 pt] node[left]{$\Omega^2=15$};

\end{tikzpicture}
\caption[La regla uniforme con varios bienes]{\label{uniformevariosbienes}\small{\textsf{\textbf{The multidimensional uniform rule}. Here, $|N|= 3$ and $|L|= 2,$ $\Omega^1=12$, $\Omega^2=15$, and $R \in \mathcal{R}^3$ is such that $p^1(R)=(2,4,8)$ and $p^2(R)=(2,7,4)$. Thus, there is not enough of the first commodity and there is too much of the second. Then, an upper bound $\lambda^1 \in \mathbb{R}_+$ and a lower bound  $\lambda^2 \in \mathbb{R}_+$ are chosen and each agent maximizes her preference in the rectangle $[0, \lambda^1]\times[\lambda^2, \Omega^2]$, so the multidimensional uniform rule allocates $u^1(R)=(2,4,6)$ and $u^2(R)=(4,7,4)$. 
}}}
\end{center}
\end{figure}

The multidimensional uniform rule is defined by applying the uniform rule commodity-wise. It determines bounds that, in turn, define a common ``budget set'' within which each agent maximizes her preference. These bounds, which geometrically form an $|L|$-dimensional box, are chosen to ensure feasibility (see Figure \ref{uniformevariosbienes}). The multidimensional uniform rule can also be seen as a special case of multidimensional sequential rule. It uses the egalitarian allocation as the initial reference point and, at each step of the iterative process, treats agents in a ``non-discriminatory'' manner.

Given $\ell \in L$, $i \in N$, and $R_i \in \mathcal{R}$, let $p^\ell(R_i)$ be the $\ell$-projection of agent $i$'s peak vector $p(R_i)$.

\vspace{5 pt}
\noindent
\textbf{Multidimensional uniform rule, $\boldsymbol{u}$:} For each $R \in \mathcal{R}^n,$ each  $i \in N,$ and each  $\ell \in L,$
 $$u_i^\ell(R)\equiv\left\{ \begin{array}{l l}
 \min\{p^\ell(R_{i}), \lambda^\ell(R)\} & \mathrm{if} \ \ \sum_{j \in N} p^\ell(R_{j}) \geq \Omega^\ell \\
\max\{p^\ell(R_{i}), \lambda^\ell(R)\} & \mathrm{if} \ \ \sum_{j \in N} p^\ell(R_{j}) \leq \Omega^\ell
\end{array} \right.$$ where $\lambda^\ell(R) \ge 0$ and solves  $\sum_{j \in N}u_j^\ell(R)=\Omega^\ell$ for each $\ell \in L.$
\vspace{5 pt}

This rule is \emph{strategy-proof} \citep[see][]{amoros2002single}, \emph{same-sided} (and therefore \emph{unanimous}), \emph{own-peak-only} and \emph{multidimensional replacement monotonic.} However, it is not \emph{efficient.} This is shown in the following example.

\begin{example}\label{ej} {\em \cite[Adapted from Example 5 in][]{Morimoto13}}
Let be $N=L=\{1,2\}$ and $\Omega = (18,12).$ Let $R_1\in \mathcal{R}$ be such that $p(R_1)=(13.5,9)$ and $(7.5,7.5) \ P_1 \ (9,6)$, and let $R_2\in \mathcal{R}$ be such that $p(R_2)=(12, 10.5)$ and $(10.5, 4.5) \ P_2 \ (9,6)$. Then,  $u_i(R)=\frac{\Omega}{2}=(9,6)$ for $i \in \{1,2\}$. Consider the feasible allocation $x=(x_1^1,x_1^2,x_2^1,x_2^2)=(7.5,7.5,10.5,4.5)$. Therefore, $x_i \ P_i \ \frac{\Omega}{2}$ for $i \in \{1,2\}$. Hence,   the multidimensional uniform rule is not \emph{efficient.} \hfill $\Diamond$
\end{example}

\begin{figure}[ht]
\begin{center}

\begin{tikzpicture}

\draw[thick,->, ] node[left, black]{$O_1$} (0,0) -- (10,0) ;
\draw[thick,->,] (0,0) -- (0,7);

\draw[thick,<-, ] (-1,6) -- (9,6) node[right, black]{$O_2$};
\draw[thick,<-,] (9,-1) -- (9,6);

\begin{scope}
        \clip (6.75,4.5) ellipse (3.76cm and 1.88cm);
        \fill[blue!10] (3,0.75) ellipse (1.875cm and 3.75cm);
    \end{scope}

\draw[
dashed] (4.5,0) -- (4.5,6);
\draw[
dashed] (0,3) -- (9,3);

\begin{scope}
        \clip (0,0) rectangle (9,6);
\draw[blue, thick] (6.75,4.5) ellipse (3.76cm and 1.88cm);
\draw[blue, thick] (6.75,4.5) ellipse (1.88cm and 0.94cm);
\draw[blue, thick] (6.75,4.5) ellipse (7.52cm and 3.76cm);


\draw[red, thick] (3,0.75) ellipse (1.875cm and 3.75cm);
\draw[red, thick] (3,0.75) ellipse (0.937cm and 1.875cm);
\draw[red, thick] (3,0.75) ellipse (3.75cm and 7.5cm);
\end{scope}

\draw[fill=blue] (6.75,4.5) circle[radius=1.5 pt] node[right]{\textcolor{blue}{$p(R_1)$}};
\draw[fill=red] (3,0.75) circle[radius=1.5 pt] node[below]{\textcolor{red}{$p(R_2)$}};
\draw[fill=black] (4.5,3) circle[radius=1.5 pt] node[below right]{$\frac{\Omega}{2}$};
\draw[fill=black] (3.75,3.75) circle[radius=1.5 pt] node[above left]{$x$};

\end{tikzpicture}
\caption[La regla uniforme con varios bienes no es \emph{eficiente}]
{\label{unifnoefic}\small{\textsf{\textbf{The multidimensional uniform rule is not efficient}. Here, both peaks are greater than equal division for each good, so the uniform rule assigns the egalitarian allocation. However,  all the allocations in the shaded area Pareto dominate it.}}}
\end{center}
\end{figure}

Example \ref{ej} is illustrated in Figure \ref{unifnoefic}. In this Edgeworth box preferences are specified so that: (i) make the uniform rule to recommend the egalitarian allocation, and  (ii) there are feasible allocations that Pareto-dominate the egalitarian allocation.

The first result concerning the multidimensional uniform rule is consequence of the rule being \emph{strategy-proof,}  \emph{unanimous} and \emph{multidimensional replacement monotonic}, together with Theorem \ref{main}.

\begin{corollary}\label{Ufeenm}
The multidimensional uniform rule is  \emph{strongly Pareto-undominated strategy-proof.}
\end{corollary}

Next, we present a new characterization of the multidimensional uniform rule.

\begin{theorem}\label{caractvariosbienes}
The multidimensional uniform rule is the only \emph{Pareto-undominated strategy-proof} rule that satisfies \emph{equal treatment} and  \emph{multidimensional replacement monotonicity.}
\end{theorem}
\begin{proof} See Appendix \ref{prueba teo 3}.
\end{proof}

\begin{remark}\emph{
The characterization of Theorem \ref{caractvariosbienes} is also valid replacing \emph{equal treatment} for \emph{equal treatment in physical terms.}\footnote{\textbf{Equal treatment in physical terms:} For each $R \in \mathcal{R}^{N}$ and each $\{i,j\} \subset N,$ if $R_i=R_j$ then $\varphi_i(R)=\varphi_j(R).$}}
\end{remark}

\begin{remark}{\em  \citep[A comparison with Theorem 4 in][]{anno2013second}} \em 
By Remark \ref{remnonbossy}, rules in two-agent economies are always \emph{multidimensional replacement monotonic}. Therefore, Theorem 4 in \cite{anno2013second} can be seen as a particular case of Theorem \ref{caractvariosbienes}.
\end{remark}

\begin{remark}{\em \citep[A comparison with Corollary 1 in][]{Morimoto13}} \em  At the expense of asking for the stronger \emph{multidimensional replacement monotonicity} instead of \emph{non-bossiness} and paying attention only to \emph{strategy-proof} rules that also fulfill the requirement of \emph{own-peak-onliness,} we get a better understanding of the  behavior of the multidimensional uniform rule in terms of efficiency: it is not only \emph{unanimous,} it is undominated (in welfare terms) among all other \emph{strategy-proof} and \emph{own-peak-only} rules.
\end{remark}

\begin{remark}{\emph{(Independence of axioms)} \em The proportional rule\footnote{\textbf{Proportional rule, $\boldsymbol{Pro}$:} For each  $R \in \mathcal{R}^{N},$ each $i \in N,$ and each $\ell \in L$,  $$Pro_i^\ell(R)\equiv\left\{ \begin{array}{l l}
 \frac{p^\ell(R_i)}{\sum_{j \in N} p^\ell(R_j)} \Omega^\ell & \text{if} \ \ \sum_{i \in N} p^\ell(R_{j}) >0, \\
\frac{\Omega^\ell}{n} & \mathrm{otherwise.}\\
\end{array} \right.$$} satisfies \emph{equal treatment} and   \emph{multidimensional replacement monotonicity,} but is not \emph{strategy-proof} and, therefore, is not \emph{Pareto-undominated strategy-proof.} The serial rule\footnote{\textbf{Serial rule,  $\boldsymbol{S}$:} For each $R \in \mathcal{R}^{N}$ and each $\ell \in L,$ if $i \in N \setminus \{n\},$ $S_i^\ell(R)\equiv \min\left\{p^\ell(R_i), \Omega^\ell-
\sum_{j<i}S_j^\ell(R)\right\},$ and $S_n^\ell(R) \equiv  \Omega^\ell-
\sum_{j<n} S_j^\ell(R).$} is  \emph{(strongly) Pareto-undominated strategy-proof} and \emph{multidimensional replacement monotonic},  but does not satisfy  \emph{equal treatment.} 
It is an open question whether  \emph{multidimensional replacement monotonicity} can be eliminated from  the characterization, or at least be weakened  to \emph{non-bossiness.} Anyway, getting rid of  \emph{non-bossiness} is already a non-trivial issue (see Remark  2 in \cite{Morimoto13} and the discussion thereafter about this point).
}
\end{remark}

\section{Final Comments}\label{ComentUnifVarios}

Some final remarks are in order. The property of \emph{multidimensional replacement monotonicity} has been fundamental in extending results from two-agent economies to those with an arbitrary number of agents---particularly in proving Lemma \ref{egregium} in Appendix \ref{appendix}. However, this monotonicity property is quite strong, and its economic interpretation remains unclear. Nevertheless, we believe its use is justified, at least provisionally, as a means to better understand the structure of \emph{strategy-proof} rules. Indeed, our most significant contribution is providing an affirmative answer to the question posed by \cite{anno2013second} regarding whether their two-agent characterization of the multidimensional uniform rule (their Theorem 4) extends to economies with an arbitrary number of agents.

Regarding the multidimensional uniform rule, \cite{anno2013second} offer an alternative characterization in the two-agent case: they assert that \emph{Pareto-undominated strategy-proofness} and the \emph{egalitarian lower bound} uniquely identify this rule. This result follows as a corollary of their Theorem 4, since in two-agent economies the \emph{egalitarian lower bound} implies \emph{equal treatment}. However, as noted in Remark \ref{remark egalitarian}, this characterization no longer holds when there are more than two agents, as any sequential rule that uses the egalitarian allocation as the initial reference point satisfies both properties.

\bibliographystyle{ecta}
\bibliography{library}

\begin{thebibliography}{19}
\newcommand{\enquote}[1]{``#1''}
\expandafter\ifx\csname natexlab\endcsname\relax\def\natexlab#1{#1}\fi

\bibitem[\protect\citeauthoryear{Amor{\'o}s}{Amor{\'o}s}{2002}]{amoros2002single}
\textsc{Amor{\'o}s, P.} (2002): \enquote{Single-peaked preferences with several commodities,} \emph{Social Choice and Welfare}, 19, 57--67.

\bibitem[\protect\citeauthoryear{Anno and Sasaki}{Anno and Sasaki}{2013}]{anno2013second}
\textsc{Anno, H. and H.~Sasaki} (2013): \enquote{Second-best efficiency of allocation rules: strategy-proofness and single-peaked preferences with multiple commodities,} \emph{Economic Theory}, 54, 693--716.

\bibitem[\protect\citeauthoryear{Arribillaga and Bonifacio}{Arribillaga and Bonifacio}{2025}]{arribillaga2025not}
\textsc{Arribillaga, R.~P. and A.~G. Bonifacio} (2025): \enquote{Not obviously manipulable allotment rules,} \emph{Economic Theory}, to appear.

\bibitem[\protect\citeauthoryear{Barber{\`a}, Jackson, and Neme}{Barber{\`a} et~al.}{1997}]{barbera1997strategy}
\textsc{Barber{\`a}, S., M.~O. Jackson, and A.~Neme} (1997): \enquote{Strategy-proof allotment rules,} \emph{Games and Economic Behavior}, 18, 1--21.

\bibitem[\protect\citeauthoryear{Ching}{Ching}{1994}]{Ching94}
\textsc{Ching, S.} (1994): \enquote{An alternative characterization of the uniform rule,} \emph{Social Choice and Welfare}, 11, 131--136.

\bibitem[\protect\citeauthoryear{Cho and Thomson}{Cho and Thomson}{2023}]{cho2023strategy}
\textsc{Cho, W.~J. and W.~Thomson} (2023): \enquote{Strategy-proofness in private good economies with linear preferences: an impossibility result,} \emph{Games and Economic Behavior}, 142, 1012--1017.

\bibitem[\protect\citeauthoryear{Goswami, Mitra, and Sen}{Goswami et~al.}{2014}]{Goswami14}
\textsc{Goswami, M.~P., M.~Mitra, and A.~Sen} (2014): \enquote{Strategy proofness and Pareto efficiency in quasilinear exchange economies,} \emph{Theoretical Economics}, 9, 361--381.

\bibitem[\protect\citeauthoryear{Hurwicz}{Hurwicz}{1972}]{Hurwicz72}
\textsc{Hurwicz, L.} (1972): \enquote{On informationally decentralized systems,} in \emph{Decision and Organization}, ed. by C.~B. McGuire and R.~Radner, North-Holland, 297--336.

\bibitem[\protect\citeauthoryear{Ju}{Ju}{2003}]{Ju03}
\textsc{Ju, B.-G.} (2003): \enquote{Strategy-proofness versus efficiency in exchange economies: general domain properties and applications,} \emph{Social Choice and Welfare}, 21, 73--93.

\bibitem[\protect\citeauthoryear{Li and Xue}{Li and Xue}{2013}]{Li13}
\textsc{Li, J. and J.~Xue} (2013): \enquote{Egalitarian division under Leontief preferences,} \emph{Economic Theory}, 54, 597--622.

\bibitem[\protect\citeauthoryear{Morimoto, Serizawa, and Ching}{Morimoto et~al.}{2013}]{Morimoto13}
\textsc{Morimoto, S., S.~Serizawa, and S.~Ching} (2013): \enquote{A characterization of the uniform rule with several commodities and agents,} \emph{Social Choice and Welfare}, 1--41.

\bibitem[\protect\citeauthoryear{Nicol{\'o}}{Nicol{\'o}}{2004}]{Nicolo04}
\textsc{Nicol{\'o}, A.} (2004): \enquote{Efficiency and truthfulness with Leontief preferences. A note on two-agent, two-good economies,} \emph{Review of Economic Design}, 8, 373--382.

\bibitem[\protect\citeauthoryear{Satterthwaite and Sonnenschein}{Satterthwaite and Sonnenschein}{1981}]{Satterthwaite81}
\textsc{Satterthwaite, M.~A. and H.~Sonnenschein} (1981): \enquote{Strategy-proof allocation mechanisms at differentiable points,} \emph{Review of Economic Studies}, 48, 587--597.

\bibitem[\protect\citeauthoryear{Schummer}{Schummer}{1997}]{Schummer97}
\textsc{Schummer, J.} (1997): \enquote{Strategy-proofness versus efficiency on restricted domains of exchange economies,} \emph{Social Choice and Welfare}, 14, 47--56.

\bibitem[\protect\citeauthoryear{Serizawa}{Serizawa}{2002}]{Serizawa02}
\textsc{Serizawa, S.} (2002): \enquote{Inefficiency of strategy-proof rules for pure exchange economies,} \emph{Journal of Economic Theory}, 106, 219--241.

\bibitem[\protect\citeauthoryear{Sprumont}{Sprumont}{1991}]{Sprumont91}
\textsc{Sprumont, Y.} (1991): \enquote{The division problem with single-peaked preferences: a characterization of the uniform allocation rule,} \emph{Econometrica}, 59, 509--519.

\bibitem[\protect\citeauthoryear{Thomson}{Thomson}{2013}]{Thomson13}
\textsc{Thomson, W.} (2013): \enquote{The Theory of Fair Allocation,} Book Manuscript, University of Rochester.

\bibitem[\protect\citeauthoryear{Thomson}{Thomson}{2016}]{thomson2016non}
---\hspace{-.1pt}---\hspace{-.1pt}--- (2016): \enquote{Non-bossiness,} \emph{Social Choice and Welfare}, 47, 665--696.

\bibitem[\protect\citeauthoryear{Troyan and Morrill}{Troyan and Morrill}{2020}]{troyan2020obvious}
\textsc{Troyan, P. and T.~Morrill} (2020): \enquote{Obvious manipulations,} \emph{Journal of Economic Theory}, 185, 104970.

\end{thebibliography}

\appendix

\section{Proofs} \label{appendix}

\subsection{Previous results}\label{resultados previos}

First, we present several known results that will be used in what follows. The first one is a result due to \cite{Morimoto13}, and states that for    \emph{strategy-proof} and \emph{non-bossy} rules,   \emph{unanimity} is equivalent to \emph{same-sidedness.}

\begin{lemma}{\em \citep{Morimoto13}}\label{SS} Every  \emph{strategy-proof,} \emph{unanimous,} and \emph{non-bossy} rule is  \emph{same-sided.}
\end{lemma}
\begin{proof} See Lemma 1 in \cite{Morimoto13}.
\end{proof}

The next result, due to \cite{amoros2002single}, states that given three allocations $\widehat{x}_i, \overline{x}_i, \widetilde{x}_i \in X,$ if at least one of the coordinates of $\overline{x}_i$  is not between the coordinates of $\widehat{x}_i$ and   $\widetilde{x}_i,$ then  $\widehat{x}_i$ can be considered as the peak of a preference relation in which   $\widetilde{x}_i$ is preferred to $\overline{x}_i.$

\begin{lemma}[\normalfont{Amorós, 2002}]\label{pref} Let $i \in N$ and $\widehat{x}_i, \overline{x}_i, \widetilde{x}_i\in X.$ If $\overline{x}_i$ is not between $\widehat{x}_i$ and $\widetilde{x}_i$, 
then there is $R_i \in \mathcal{R}$ such that $p(R_i)=\widehat{x}_i$ and $\widetilde{x}_i \ P_i \ \overline{x}_i.$
\end{lemma}
\begin{proof}
See Lemma  1 in \cite{amoros2002single}.
\end{proof}

Let  be $\varphi \in \Phi$ and $i \in N.$ For each  $R_{-i} \in \mathcal{R}^{n-1}$, define the  \textbf{option set of agent   $\boldsymbol{i}$ left open by $\boldsymbol{R_{-i}}$ under $\boldsymbol{\varphi}$} by $$\boldsymbol{O_i^{\varphi}(R_{-i})}\equiv\{x_i \in X \ | \ \text{there is} \ R_i \in \mathcal{R} \ \text{such that} \ \varphi_i(R_i, R_{-i})=x_i\}.$$ 
For each  $i \in N,$ each $R_i \in \mathcal{R},$ and each $Y \subseteq X$,  define the  \textbf{choice set of agent $\boldsymbol{i}$ on $\boldsymbol{Y}$ with respect to  $\boldsymbol{R_i}$} by $$\boldsymbol{C_i(R_i, Y)}\equiv\{x_i \in Y \ |  \text{ for each } y_i \in Y, x_i \ R_i  \ y_i\}.$$

\begin{lemma}\label{choiceyoption} Let $\varphi$ be a rule. Then,
\begin{enumerate}[(i)]
    \item $\varphi$ is \emph{strategy-proof} if and only if for each $R \in \mathcal{R}$ and each $i \in N,$  $\varphi_i(R) \in C_i(R_i,O_i^{\varphi}(R_{-i})).$

    \item if $\varphi \in \Phi^{\star},$ then for each $i \in N,$ each $R_{-i} \in \mathcal{R}^{n-1},$ and each
$\ell \in L,$ there are $a^\ell, b^\ell \in [0, \Omega^\ell]$ such that $O_i^{\varphi}(R_{-i})=\prod_{\ell \in L} [a^\ell,b^\ell].$

    \item  if $\varphi \in \Phi^{\star},$ then for each $R \in \mathcal{R}$ and each $i \in N,$ $C_i(R_i, O_i^{\varphi}(R_{-i}))=\{\varphi_i(R)\}.$
\end{enumerate}
\end{lemma}
\begin{proof} Part (i) is straightforward and Part (ii) is Lemma 8  in \cite{anno2013second}. To see Part  (iii), notice that Part (ii) and single-peakedness imply  $|C_i(R_i, O_i^{\varphi}(R_{-i}))|=1,$ so the result follows from  Part (i).
\end{proof}

\noindent Since option sets of  \emph{strategy-proof} and \emph{own-peak-only} rules are singletons, we often abuse notation and write, for each $\varphi \in \Phi^{\star},$ each $R \in \mathcal{R}$, and each $i \in N,$ $C_i(R_i, O_i^{\varphi}(R_{-i}))=\varphi_i(R).$

Next, we show that domination between rules can easily be translated into the inclusion of option sets.

\begin{lemma}\label{options}
Let $\varphi$ and  $\psi$ be two rules in $\Phi^{\star}$. Then,  $\varphi \succcurlyeq \psi$ if and only if for each  $i \in N$ and each $R_{-i} \in \mathcal{R}^{n-1},$ $O_i^{\psi}(R_{-i}) \subseteq O_i^{\varphi}(R_{-i}).$
\end{lemma}
\begin{proof}
($\Longrightarrow$) Let $i \in N$ and $R_{-i} \in \mathcal{R}^{n-1}.$ Assume $x_i \in O_i^{\psi}(R_{-i})$. Take $R_i \in \mathcal{R}$ such that $p(R_i)=x_i.$ By Lemma \ref{choiceyoption} (iii), $\psi_i(R)=x_i.$ As $\varphi \succcurlyeq \psi,$ we have  $\varphi_i(R) \ R_i \ \psi_i(R)=p(R_i)$ and, therefore, $\varphi_i(R)=x_i$. Thus, $x_i \in O_i^{\varphi}(R_{-i}).$

\noindent
($\Longleftarrow$)  Let $i \in N$ and $R \in \mathcal{R}^{n}.$ By Lemma \ref{choiceyoption} (i), $\varphi_i(R) \in C_i(R_i, O_i^{\varphi}(R_{-i}))$ and  $\psi_i(R) \in C_i(R_i, O_i^{\psi}(R_{-i})).$ As $O_i^{\psi}(R_{-i}) \subseteq O_i^{\varphi}(R_{-i}),$ $\varphi_i(R) \ R_i \ \psi_i(R).$
\end{proof}

\subsection{Proof of Theorem \ref{main}}\label{prueba teo 1}

Before proving the theorem, we need some technical results. Lemmata \ref{Auxiliar peaks} and \ref{egregium} are used to prove that every \emph{strategy-proof,} \emph{unanimous} and \emph{multidimensional replacement monotonic} rule is \emph{own-peak-only} (Lemma \ref{peakonly}). The proof of Theorem \ref{main} then follows from Lemma \ref{peakonly} and the results in \ref{resultados previos}. 

\begin{lemma}\label{Auxiliar peaks} Let $\varphi$ be a   \emph{strategy-proof} and \emph{non-bossy} rule. Then, for each  $R \in \mathcal{R}^{n},$ each $S\subseteq N,$ and each  $R^{\star}_S \in \mathcal{R}^{|S|}$ such that $p(R_j^{\star})=\varphi_j(R)$ for each $j \in S$, we have $\varphi_j\left(R_S^{\star}, R_{-S}\right)=\varphi_j(R).$
\end{lemma}
\begin{proof} Let $R \in \mathcal{R}^{n},$ $S \subseteq N,$ $j \in S$, and  $R^{\star}_S \in \mathcal{R}^{|S|}$ be such that $p(R_j^{\star})=\varphi_j(R)$ for each $j \in S$. Let $i \in S$. Since $\varphi$ is \emph{strategy-proof,} $\varphi_i(R_i^{\star}, R_{-i})=\varphi_i(R)$ (otherwise agent $i$ gets her peak in economy $(R_i^{\star}, R_{-i})$ by  declaring $R_i$). By \emph{non-bossiness,} $\varphi(R_i^{\star}, R_{-i})=\varphi(R).$  Let $k \in S\setminus \{ i\}$. Then, by  \emph{strategy-proofness,} $\varphi_k(R_{i,k}^{\star}, R_{-i,k})=\varphi_k(R_i^{\star}, R_{-i})$ and, by \emph{non-bossiness,} $\varphi(R_{i,k}^{\star}, R_{-i,k})=\varphi(R_i^{\star}, R_{-i})=\varphi(R).$ Continuing in the same fashion, changing the preference of each remaining agent in $S$ one at a time, the result follows.
\end{proof}

Given $i \in N,$ preferences $R_i, \widetilde{R}_i \in \mathcal{R}$ and $x_i \in X,$ define $\boldsymbol{L(R_i, \widetilde{R}_i, x_i)}$  as the set of  
$\ell \in L$ that satisfy one of the following: 
\begin{enumerate}[(i)]
    \item $p^\ell(R_i) < x_i^\ell \ \text{and} \ p^\ell(\widetilde{R}_i) \leq x_i^\ell,$
    \item $p^\ell(R_i) > x_i^\ell \ \text{and} \ p^\ell(\widetilde{R}_i) \geq x_i^\ell,$
    \item $p^\ell(R_i)= p^\ell(\widetilde{R}_i)=x_i^\ell.$
\end{enumerate}

\begin{lemma}\label{egregium}
Let  $\varphi$ be a   \emph{strategy-proof,} \emph{unanimous} and \emph{multidimensional replacement monotonic} rule. Let $R \in \mathcal{R}^{n},$ $i \in N,$ and $j \in N\setminus\{i\}$. Then, $$\varphi_i^\ell(\widetilde{R}_i, R_{j}^{\star}, R_{-i,j})=\varphi_i^\ell(R)$$
for each $\widetilde{R}_i \in \mathcal{R},$ each $\ell \in L(R_i, \widetilde{R}_i,\varphi_i(R)),$ and each $R_j^{\star} \in \mathcal{R}$ such that $p(R_j^{\star})=\varphi_j(R).$
\end{lemma}
\begin{proof} Let $\varphi$ be a rule that satisfies the properties listed in the lemma.  By Lemma \ref{SS}, $\varphi$ is \emph{same-sided,} and by Remark \ref{remnonbossy}, $\varphi$ is  \emph{non-bossy.} Let $R \in \mathcal{R}^{n},$  $i \in N$, and consider the profile $R_{-i}^{\star} \in \mathcal{R}^{n-1}$ such that $p(R_j^{\star})=\varphi_j(R)$ for each $j \in N\setminus \{i\}.$ We will prove the lemma in several steps.

\medskip 

\noindent \textbf{Step 1: For each $\boldsymbol{i \in N}$ and each $\boldsymbol{R_i' \in \mathcal{R}}$ such that $\boldsymbol{p(R'_i)=p(R_i)},$  $$\boldsymbol{\varphi(R_i', R^{\star}_{-i})=\varphi(R).}$$}
By Lemma \ref{Auxiliar peaks} and \emph{non-bossiness,} it is sufficient to see that  $\varphi_i(R_i', R_{-i}^{\star})=\varphi_i(R_i, R_{-i}^{\star}).$   
Let $\ell \in L.$ If  $\varphi_i^\ell(R_i', R_{-i}^{\star}) \leq p^\ell(R_i'),$ by \emph{same-sidedness} and Lemma \ref{Auxiliar peaks} we have  $\varphi_j^\ell(R_i', R_{-i}^{\star}) \leq p(R_{j}^{\star})=\varphi_j^\ell(R_i, R_{-i}^{\star})$ for each $j \in N\setminus \{i\}.$
Therefore, $\varphi_i^\ell(R_i', R_{-i}^{\star})=\Omega^\ell-\sum_{j \in N\setminus \{i\}} \varphi_j^\ell(R_i', R_{-i}^{\star}) \geq \Omega^\ell - \sum_{j \in N\setminus \{i\}} \varphi_j^\ell(R_i,R_{-i}^{\star})=\varphi_i^\ell(R_i, R_{-i}^{\star}).$ In consequence, as  $p^\ell(R_i')=p^\ell(R_i),$
\begin{equation}\label{equuno}
\varphi_i^\ell(R_i, R_{-i}^{\star}) \leq \varphi_i^\ell(R_i', R_{-i}^{\star}) \leq p^\ell(R_i).
\end{equation}
Analogously, we can show that if $\varphi_i^\ell(R_i',R_{-i}^{\star}) \geq p^\ell(R_i'),$ then
\begin{equation}\label{equdo}
\varphi_i^\ell(R_i, R_{-i}^{\star}) \geq \varphi_i^\ell(R_i', R_{-i}^{\star}) \geq p^\ell(R_i).
\end{equation}
Since both  (\ref{equuno}) and (\ref{equdo}) are true for each $\ell \in L,$ if  $\varphi_i(R_i',R_{-i}^{\star})\neq\varphi_i(R_i, R_{-i}^{\star})$ then  $$\varphi_i(R_i',R_{-i}^{\star}) \ P_i \ \varphi_i(R_i, R_{-i}^{\star}),$$  violating the  \emph{strategy-proofness} of $\varphi.$ Hence $\varphi_i(R_i',R_{-i}^{\star})= \varphi_i(R_i, R_{-i}^{\star})$, as desired. 

\medskip

\noindent \textbf{Step 2: For each $\boldsymbol{i \in N},$ each $\boldsymbol{\widetilde{R}_i \in \mathcal{R}}$, and each $\boldsymbol{\ell \in L(R_i, \widetilde{R}_i,\varphi_i(R))},$ $$\boldsymbol{\varphi_i^\ell(\widetilde{R}_i, R_{-i}^{\star})=\varphi_i^\ell(R).}$$} Notice that, by Lemma \ref{Auxiliar peaks}, it is sufficient to see that $\varphi_i^\ell(\widetilde{R}_i, R_{-i}^{\star})=\varphi_i^\ell(R_i, R_{-i}^{\star}).$ Assume this is not true. Let us analyze the case in which   $\ell \in L$  is such that $p^\ell(\widetilde{R}_i)+\sum_{j \in N\setminus \{i\}}p^\ell(R_j^{\star})\leq\Omega^\ell,$ since an analogous reasoning applies to the symmetric case. We have that $p^\ell(\widetilde{R}_i)\leq \Omega^\ell-\sum_{j \in N\setminus \{i\}}p^\ell(R_j^{\star})=\varphi_i^\ell(R_i, R_{-i}^{\star}),$ since, for each $j \in N\setminus \{i\},$ $p^\ell(R_j^{\star})=\varphi_j^\ell(R_i, R_{-i}^{\star}).$ 
Therefore, as $\varphi_i^\ell(R_i, R_{-i}^{\star})=\varphi_i^\ell(R)$ and $\ell \in L(R_i, \widetilde{R}_i,\varphi_i(R)),$ we have
\begin{equation}\label{eq lemma 9 - 1}
    p^\ell(R_i) < \varphi_i^\ell(R_i, R_{-i}^{\star}).
\end{equation}
By  \emph{same-sidedness,}  $\varphi_j^\ell(R_i, R_{-i}^{\star})\geq p^\ell(R_j^{\star})$ for each $j \in N\setminus\{i\}$. Then,  $\sum_{j \in N \setminus\{i\}}\varphi_j^\ell(\widetilde{R}_i, R_{-i}^{\star})\geq \sum_{j \in N \setminus\{i\}}\varphi_j^\ell(R_i, R_{-i}^{\star})$ and, by feasibility and the fact that $\varphi_i^\ell(\widetilde{R}_i, R_{-i}^{\star})\neq \varphi_i^\ell(R_i, R_{-i}^{\star}),$
\begin{equation}\label{eq lemma 9 - 2}
    \varphi_i^\ell(\widetilde{R}_i, R_{-i}^{\star})<\varphi_i^\ell(R_i, R_{-i}^{\star}).
\end{equation}
In consequence, by \eqref{eq lemma 9 - 1}, \eqref{eq lemma 9 - 2}, and  Lemma \ref{pref}, there is $R_i' \in \mathcal{R}$ with $p(R_i')=p(R_i)$ such that $\varphi_i^\ell(\widetilde{R}_i, R_{-i}^{\star})P_i'\varphi_i^\ell(R_i, R_{-i}^{\star}).$ By Step 1, $\varphi_i^\ell(R_i', R_{-i}^{\star})=\varphi_i^\ell(R_i, R_{-i}^{\star}).$ Thus,    $\varphi_i^\ell(\widetilde{R}_i, R_{-i}^{\star})P_i'\varphi_i^\ell(R_i', R_{-i}^{\star}),$ contradicting the \emph{strategy-proofness} of $\varphi.$ Hence, $\varphi_i^\ell(\widetilde{R}_i, R_{-i}^{\star})=\varphi_i^\ell(R_i, R_{-i}^{\star})$ and the result follows.

\medskip

\noindent \textbf{Step 3: For each  $\boldsymbol{i \in N},$ each $\boldsymbol{\widetilde{R}_i \in \mathcal{R}},$ each $\boldsymbol{\ell \in L(R_i, \widetilde{R}_i,\varphi_i(R))},$ each $\boldsymbol{j \in N\setminus\{i\}}$, and each $\boldsymbol{R_j^{\star} \in \mathcal{R}}$ such that $\boldsymbol{p(R_j^{\star})=\varphi_j(R)},$  
$$\boldsymbol{\varphi_i^\ell(\widetilde{R}_i, R_{j}^{\star}, R_{-i,j})=\varphi_i^\ell(R)}.$$}
By Step 2, we only need to show that $\varphi_i^\ell(\widetilde{R}_i, R_{j}^{\star}, R_{-i,j})=\varphi_i^\ell(\widetilde{R}_i, R_{-i}^{\star})$. First, we show that $\varphi_i^\ell(\widetilde{R}_i, R_{-i,j}^{\star}, R_{j})=\varphi_i^\ell(\widetilde{R}_i, R_{-i}^{\star}).$ By  \emph{non-bossiness,} it is sufficient to see that $\varphi_j^\ell(\widetilde{R}_i, R_{-i,j}^{\star}, R_{j})=\varphi_j^\ell(\widetilde{R}_i, R_{-i}^{\star}).$ Assume, without loss of generality, that $\varphi_j^\ell(\widetilde{R}_i, R_{-i,j}^{\star}, R_{j})<\varphi_j^\ell(\widetilde{R}_i, R_{-i}^{\star}).$ This means, by  \emph{multidimensional replacement monotonicity,} that
\begin{equation}\label{NN1}
\varphi_k^\ell(\widetilde{R}_i, R_{-i,j}^{\star}, R_{j})\geq \varphi_k^\ell(\widetilde{R}_i, R_{-i}^{\star}) \text{ for each } k \in N\setminus\{j\}.
\end{equation}
If  $\varphi_i^\ell(\widetilde{R}_i, R_{-i,j}^{\star}, R_{j})=\varphi_i^\ell(\widetilde{R}_i, R_{-i}^{\star})$ we get the result. If not, since from  Lemma \ref{Auxiliar peaks} and Step 2, $\varphi_i^\ell(\widetilde{R}_i, R_{-i}^{\star})=\varphi_i^\ell( R)=\varphi_i^\ell(R_{-i,j}^{\star}, R_{i,j}),$ we have $\varphi_i^\ell(\widetilde{R}_i, R_{-i,j}^{\star}, R_{j})>\varphi_i^\ell(R_{-i,j}^{\star}, R_{i,j})$ which implies, by  \emph{multidimensional replacement monotonicity,} that
\begin{equation}\label{NN2}
\varphi_k^\ell(\widetilde{R}_i, R_{-i,j}^{\star}, R_{j})\leq \varphi_k^\ell(R_{-i,j}^{\star}, R_{i,j}) \text{ para cada }k \in N\setminus \{i,j\}.
\end{equation}
Take $k \in N\setminus\{i,j\}.$ By Lemma  \ref{Auxiliar peaks}, Step 2, and   \emph{multidimensional replacement monotonicity,}  $\varphi_k^\ell(R_{-i,j}^{\star}, R_{i,j})=\varphi_k^\ell(R)=\varphi_k^\ell(\widetilde{R}_i, R_{-i}^{\star}).$ In consequence, by  (\ref{NN1}) and (\ref{NN2}) we have $\varphi_k^\ell(\widetilde{R}_i, R_{-i,j}^{\star}, R_{j})=\varphi_k^\ell(\widetilde{R}_i, R_{-i}^{\star}).$ By  \emph{multidimensional replacement monotonicity,} $\varphi_j^\ell(\widetilde{R}_i, R_{-i,j}^{\star}, R_{j})=\varphi_j^\ell(\widetilde{R}_i, R_{-i}^{\star}),$ a contradiction. Continuing in the same fashion, changing one preference at a time, we can prove that for each  $S\subseteq N\setminus\{i\},$  $\varphi_j^\ell(\widetilde{R}_i, R_{-S}^{\star}, R_{S})=\varphi_j^\ell(\widetilde{R}_i, R_{-i}^{\star}).$  We get the result considering   $S=N\setminus \{i,j\}.$
\end{proof}

\begin{lemma}\label{peakonly}
Let $\varphi$ be a   \emph{strategy-proof,} \emph{unanimous} and \emph{multidimensional replacement monotonic} rule. Then   $\varphi$ is \emph{own-peak-only.}
\end{lemma}
\begin{proof}
Let $\varphi$ be a rule that satisfies the properties listed in the lemma. By Lemma \ref{SS}, $\varphi$ is \emph{same-sided,} and by Remark  \ref{remnonbossy}, $\varphi$ is   \emph{non-bossy.} Assume that  $\varphi$ is not   \emph{own-peak-only.} Then, there are  $R \in \mathcal{R}^{n},$ $i \in N,$ $\ell \in L,$ and $R_i'\in \mathcal{R}$ with $p(R'_i)=p(R_i)$ such that, without loss of generality,  
\begin{equation}\label{prueba own-peak-only 1}
    \varphi_i^\ell(R_i',R_{-i})<\varphi_i^\ell(R).
\end{equation}
By  \emph{same-sidedness,} both $\varphi_i^\ell(R_i',R_{-i})$ and $\varphi_i^\ell(R)$ are on the same side of $p^\ell(R_i)$. Assume, again without loss of generality, that $p^\ell(R_i) \leq \varphi_i^\ell(R_i',R_{-i})<\varphi_i^\ell(R).$ By feasibility and  \emph{same-sidedness},  there is  $j \in N\setminus \{i\}$ such that $p^\ell(R_j) \leq \varphi_j^\ell(R)<\varphi_j^\ell(R_i',R_{-i}).$ Let $R_i^{\star} \in \mathcal{R}$ be such that $p(R_i^{\star})=\varphi_{i}(R_i',R_{-i})$ and let $R_j^{\star} \in \mathcal{R}$ be such that $p(R_j^{\star})=\varphi_{i}(R).$ Then, it is easily seen that  $\ell \in L(R_i, R_i^{\star}, \varphi_i(R))$ and $\ell \in L(R_j, R_j^{\star}, \varphi_j(R_i', R_{-i})).$ By Lemma \ref{egregium},  $\varphi_i^\ell(R_{i,j}^{\star}, R_{-i,j})=\varphi_i^\ell(R),$ and   $\varphi_j^\ell(R_{i,j}^{\star}, R_{-i,j})=\varphi_j^\ell(R_i',R_{-i}).$ By Lemma \ref{Auxiliar peaks}, $\varphi(R_i', R_{-i})=\varphi(R_i^{\star},R_{-i}),$ and therefore $\varphi_j^\ell(R_{i,j}^{\star}, R_{-i,j})=\varphi_j^\ell(R_i^{\star},R_{-i}).$ In consequence, \emph{non-bossiness} implies $\varphi_i^\ell(R_{i,j}^{\star}, R_{-i,j})=\varphi_i^\ell(R_i^{\star},R_{-i})=\varphi_i^\ell(R_i',R_{-i}).$ It follows that
$\varphi_i^\ell(R)=\varphi_i^\ell(R_{i,j}^{\star}, R_{-i,j})=\varphi_i^\ell(R_i',R_{-i}).$ This contradicts \eqref{prueba own-peak-only 1}. Hence, $\varphi$ is \emph{own-peak-only}.
\end{proof}

\vspace{15 pt}

\noindent
\emph{Proof of Theorem \ref{main}.} Let $\varphi$ be a rule that satisfies the properties listed in the theorem.  By Lemma \ref{peakonly}, $\varphi$ is  \emph{own-peak-only.}  Let $\psi \in \Phi^{\star}$ be such that $\psi \succcurlyeq \varphi$ and $\psi \not  \neq \varphi.$ Then, by Lemma \ref{options}, for each  $i \in N$ and each $R_{-i} \in \mathcal{R}^{n-1}$ we have $O_i^{\psi}(R_{-i}) \subseteq O_i^{\varphi}(R_{-i}).$ If for each $i \in N$ and each $R_{-i} \in \mathcal{R}^{n-1}$ we have $O_i^{\psi}(R_{-i}) = O_i^{\varphi}(R_{-i}),$ then by Lemma \ref{choiceyoption} (iii)  we have, for each $i \in N$ and each $R \in \mathcal{R}^{n},$ $\psi_i(R)=C_i(R_i,O_i^{\psi}(R_{-i}) )=C_i(R_i,O_i^{\varphi}(R_{-i}) )=\varphi_i(R),$ and thus $\psi=\varphi,$ contradicting our hypothesis. Therefore, there are $i \in N$ and $R_{-i} \in \mathcal{R}^{n-1}$ such that  $O_i^{\psi}(R_{-i}) \subsetneq O_i^{\varphi}(R_{-i}).$ It follows that there is $\widetilde{x}_i \in X$ such that $\widetilde{x}_i \in O_i^{\psi}(R_{-i})$ and $\widetilde{x}_i \notin O_i^{\varphi}(R_{-i}).$ Let $\widetilde{R}_i \in \mathcal{R}$ be such that $p(\widetilde{R}_i)=\widetilde{x}_i.$ Then,   $\widetilde{x}_i=C_i(\widetilde{R}_i,O_i^{\psi}(R_{-i}))=\psi_i(\widetilde{R}_i,R_{-i}).$ Let $y_i \equiv C_i(\widetilde{R}_i,O_i^{\varphi}(R_{-i}))=\varphi_i(\widetilde{R}_i,R_{-i}).$ Since $\widetilde{x}_i \notin O_i^{\varphi}(R_{-i}),$ $y_i \neq \widetilde{x}_i.$ Then, there is  $\ell \in L$ such that, without loss of generality,  $y_i^\ell < \widetilde{x}_i^\ell=p^\ell(\widetilde{R}_i)$. By Lemma   \ref{SS}, $\varphi$ is  \emph{same-sided,} and since  $\varphi_i^\ell(\widetilde{R}_i,R_{-i})<p^\ell(\widetilde{R}_i),$ this implies, for each $j \in N\setminus\{i\},$   $\varphi_j^\ell(\widetilde{R}_i,R_{-i})\leq p^\ell(R_j).$  Being  $\varphi_i^\ell(\widetilde{R}_i,R_{-i})<p^\ell(\widetilde{R}_i)=\psi_i^\ell(\widetilde{R}_i,R_{-i}),$ by feasibility there is $j \in N\setminus\{i\}$ such that $\varphi_j^\ell(\widetilde{R}_i,R_{-i})>\psi_j^\ell(\widetilde{R}_i,R_{-i}).$ Thus,  $$p^\ell(R_j) \geq \varphi_j^\ell(\widetilde{R}_i,R_{-i})>\psi_j^\ell(\widetilde{R}_i,R_{-i}).$$
By Lemma \ref{pref}, there is $\widetilde{R}_j \in \mathcal{R}$ such that $p(\widetilde{R}_j)=p(R_j)$ and $\varphi_j(\widetilde{R}_i,R_{-i}) \ \widetilde{P}_j \ \psi_j(\widetilde{R}_i,R_{-i}).$ Since $\psi$ is \emph{own-peak-only,} we have $$\varphi_j(\widetilde{R}_{i,j}, R_{-i,j}) \ \widetilde{P}_j \ \psi_j(\widetilde{R}_{i,j},R_{-i,j}).$$ But this last statement contradicts that   $\psi \succcurlyeq \varphi.$ Hence, if $\psi \succcurlyeq \varphi$ it follows that $\psi = \varphi$ and thus $\varphi$ is \emph{strongly Pareto-undominated strategy-proof}.  \hfill $\square $

\subsection{Proof of Lemma \ref{secvarios}}\label{prueba lema 3}

The next result presents a property of one-dimensional rules that satisfy  \emph{strategy-proofness} and  \emph{same-sidedness.} We present an easy proof for completeness.

\begin{lemma}\label{uncompromise} Let $\mathcal{R}$ be the one-dimensional single-peaked domain and let $\varphi$ be a \emph{strategy-proof} and \emph{same-sided} rule defined on that domain.  For each $R \in \mathcal{R}^{n},$ $i \in N$ and  $R_i'\in \mathcal{R},$ we have
\begin{enumerate}[(i)]
    \item if $p(R_i) < \varphi_i(R)$ and $p(R_i') \leq \varphi_i(R),$ then   $\varphi_i(R_i', R_{-i})=\varphi_i(R),$

    \item if $p(R_i) > \varphi_i(R)$ and $p(R_i') \geq \varphi_i(R),$ then   $\varphi_i(R_i', R_{-i})=\varphi_i(R).$
\end{enumerate}
\end{lemma}
\begin{proof} Let us check (i), since  (ii) is analogous. Suppose $p(R_i) < \varphi_i(R)$ and  $p(R_i') \leq \varphi_i(R).$ As $\varphi$ is \emph{same-sided,} $p(R_j) \leq \varphi_j(R)$ for each  $j \in N\setminus \{i\}.$ Then $p(R_i')+\sum_{j \in N\setminus \{i\}}p(R_j) \leq \sum_{j \in N} \varphi_j(R)=\Omega.$ Again by \emph{same-sidedness,}  $\varphi_i(R_i',R_{-i})\geq p(R_i').$ Assume  $\varphi_i(R_i',R_{-i}) \neq \varphi_i(R).$ There are two cases to analyze:

\noindent \textbf{1. $\boldsymbol{\varphi_i(R_i',R_{-i}) > \varphi_i(R)}.$} Then $\varphi_i(R_i',R_{-i}) > \varphi_i(R) \geq p(R_i'),$ which contradicts \emph{strategy-proofness} of $\varphi.$

\noindent \textbf{2.  $\boldsymbol{\varphi_i(R_i',R_{-i}) < \varphi_i(R)}.$} It is a well-known fact that a one-dimensional  \emph{strategy-proof} and  \emph{same-sided} rule is  \emph{own-peak-only} \citep[see, for example,][]{Ching94}. 
 Let $\widetilde{R}_i \in \mathcal{R}$ be such that $p(\widetilde{R}_i)=p(R_i)$ and $\varphi_i(R_i', R_{-i})  \ \widetilde{P}_i \ \varphi_i(R).$ As $\varphi$ is \emph{own-peak-only}, $\varphi_i(R)=\varphi_i(\widetilde{R}_i, R_{-i}),$ and therefore $\varphi_i(R_i', R_{-i}) \  \widetilde{P}_i \    \varphi_i(\widetilde{R}_i, R_{-i}),$ contradicting \emph{strategy-proofness.}
\end{proof}

\vspace{15 pt}

\noindent
\emph{Proof of Lemma \ref{secvarios}}. Let $\phi=\left(\phi^1, \ldots, \phi^l \right)$ be a multidimensional sequential rule. By definition of $\phi$, given $R \in \mathcal{R}^n$,  the coordinate function $\phi^\ell$ depends only on $p^\ell(R)$ for each $\ell \in L$. Notice 
 that $\phi$ inherits the properties of  \emph{replacement monotonicity} and \emph{same-sidedness} from each of its coordinate functions. It only remains to be seen that $\phi$ is also  \emph{strategy-proof}. Let be $R \in \mathcal{R}^{n},$ $i \in N$ and $R_i'\in \mathcal{R}.$ We must prove that $\phi_i(R)  \ R_i \  \phi_i(R_i',R_{-i}).$ Take  $\ell \in L$ and assume $\sum_{j \in N} p^\ell(R_j) \geq \Omega^\ell.$ By  \emph{same-sidedness,} 
\begin{equation}\label{multidimensional 1}
\phi_i^\ell(R) \leq p^\ell(R_i).    
\end{equation}
Assume $\phi_i^\ell(R) < p^\ell(R_i).$ If $p^\ell(R_i') \geq \phi_i^\ell(R)$ then, by Lemma \ref{uncompromise}, $\phi_i^\ell(R_i',R_{-i})=\phi_i^\ell(R).$ If $p^\ell(R_i') < \phi_i^\ell(R),$ we have two cases to analyze.

\noindent \textbf{1. $\boldsymbol{p^\ell(R_i')+\sum_{j \in N\setminus \{i\}} p^\ell(R_j) \geq \Omega^\ell}.$} By  \emph{same-sidedness,} $\phi_i^\ell(R_i',R_{-i})\leq p^\ell(R_i') <\phi_i^\ell(R).$

\noindent \textbf{2. $\boldsymbol{p^\ell(R_i')+\sum_{j \in N\setminus \{i\}} p^\ell(R_j) < \Omega^\ell}.$} Assume $\phi_i^\ell(R_i',R_{-i})>\phi_i^\ell(R).$ By feasibility there is $j \in N \setminus \{i\}$ such that  $\phi_j^\ell(R_i',R_{-i}) < \phi_j^\ell(R).$ By  \emph{same-sidedness}, we have $p^\ell(R_j) \leq \phi_j^\ell(R_i',R_{-i}) < \phi_j^\ell(R)\leq p^\ell(R_j),$ which is absurd. In consequence, $\phi_i^\ell(R_i',R_{-i})\leq\phi_i^\ell(R).$

We conclude that, either $\phi_i^\ell(R)=p^\ell(R_i)$ or $\phi_i^\ell(R_i',R_{-i})\leq\phi_i^\ell(R)<p^\ell(R_i).$ With an analogous reasoning we can see that, for  $\ell \in L$ such that $\sum_{j \in N} p^\ell(R_j) < \Omega^\ell,$ we have  $\phi_i^\ell(R)=p^\ell(R_i)$ ó $\phi_i^\ell(R_i',R_{-i})\geq\phi_i^\ell(R)>p^\ell(R_i).$ This implies  $\phi_i(R) \  R_i \ \phi_i(R_i',R_{-i}),$ and thus $\phi$ is \emph{strategy-proof}. \hfill $\square $

\subsection{Proof of Theorem \ref{caractvariosbienes}}\label{prueba teo 3} 

The following result is required for the proof of the theorem.

\begin{lemma}\label{lemma unanimity}
    Any \emph{Pareto-undominated strategy-proof} and \emph{multidimensional replacement monotonic} rule 
    satisfies \emph{unanimity}.
\end{lemma}
\begin{proof}
    Let $\varphi$ be a \emph{Pareto-undominated strategy-proof} and \emph{multidimensional replacement monotonic} rule. Then, $\varphi$ is \emph{own-peak-only} and since it is \emph{non-bossy}  by Remark \ref{remnonbossy}, it is also \emph{peaks-only} (see definition in Footnote \ref{footnote peaks-only}).   
    Assume $\varphi$ is not \emph{unanimous}. Then, there is $\overline{R} \in \mathcal{R}^n$ with  $\sum_{i \in N}p^\ell(\overline{R}_i)=\Omega^\ell$ for each $\ell \in L$ such that, for some $i^\star \in N$ and some $\ell^\star \in L$, $\varphi_{i^\star}(\overline{R})\neq p^{\ell^\star}(\overline{R}_{i^\star})$. Notice that this implies that $\varphi(\overline{R})\neq p(\overline{R})$.
    Next, define rule $\psi$ as follows: for each $R \in \mathcal{R}^n$, $$\psi(R)=\begin{cases}
        p(\overline{R}) & \text{ if } p(R)=p(\overline{R})\\
        \varphi(R) & \text{otherwise}
    \end{cases}$$
    Clearly, $\psi$ is \emph{peaks-only} and thus \emph{own-peak-only}.
    \medskip
    
    \noindent \textbf{Claim: $\boldsymbol{\psi}$ is \emph{strategy-proof}.} Assume this is not the case. 
    Then, there is $\widetilde{R} \in \mathcal{R}^n$, $i \in N$, and $R_i' \in \mathcal{R}$ such that \begin{equation}
        \label{prueba unanime 1} \psi_i(R_i',\widetilde{R}_{-i}) \ \widetilde{P}_i \ \psi_i(\widetilde{R}).
    \end{equation}
    First, notice that $p(\widetilde{R})\neq p(\overline{R})$. Otherwise, $p(\widetilde{R}_i)=p(\overline{R}_i)=\psi_i(\widetilde{R})$, contradicting \eqref{prueba unanime 1}. Therefore, by the definition of $\psi$, $\psi(\widetilde{R})=\varphi(\widetilde{R})$. Furthermore, $p(R_i',\widetilde{R}_{-i})=p(\overline{R})$. Otherwise $\psi(R_i',\widetilde{R}_{{-i}})=\varphi(R_i',\widetilde{R}_{-i})$ and \eqref{prueba unanime 1} becomes $\varphi_i(R_i',\widetilde{R}_{-i}) \ \widetilde{P}_i \ \varphi_i(\widetilde{R})$, contradicting that $\varphi$ is \emph{strategy-proof}. 
    Thus, $p(R_i')=p(\overline{R}_i)$, $p(\widetilde{R}_{-i})=p(\overline{R}_{-i})$, and by \emph{own-peak-onliness} and the definition of $\psi$,   \eqref{prueba unanime 1} becomes 
    \begin{equation}
        \label{prueba unanime 2} p(\overline{R}_i) \ \widetilde{P}_i \ \varphi_i( \widetilde{R}_i, \overline{R}_{-i}).
    \end{equation}
    Next, notice that 
    \begin{equation}\label{between 1}
    \varphi_i(\overline{R}) \text{ is between } p(\overline{R}_i) \text{ and } \varphi_i(\widetilde{R}_i, \overline{R}_{-i}).     
    \end{equation}
    Otherwise, by Lemma \ref{pref}, there is $\widehat{R}_i \in \mathcal{R}$ with $p(\widehat{P}_i)=p(\overline{R}_i)$ such that $\varphi_i(\widetilde{R}_i, \overline{R}_{-i}) \ \widehat{P}_i \ \varphi_i(\overline{R})$ and, by \emph{own-peak-onliness},  $\varphi_i(\widetilde{R}_i, \overline{R}_{-i}) \ \widehat{P}_i \ \varphi_i(\widehat{R}_i,\overline{R}_i)$, contradicting \emph{strategy-proofness}. Similarly, it can be shown that 
    \begin{equation}\label{between 2}
    \varphi_i(\widetilde{R}_i, \overline{R}_{-i}) \text{ is between } p(\widetilde{R}_i) \text{ and } \varphi_i(\overline{R}).     
    \end{equation}
    By \eqref{between 1} and \eqref{between 2}, it is clear that $\varphi_i(\widetilde{R}_i, \overline{R}_{-i})$ is between $p(\widetilde{R}_i)$ and $p(\overline{R}_i)$. Therefore, by multidimensional single-peakedness, $\varphi_i(\widetilde{R}_i, \overline{R}_{-i}) \ \widetilde{P}_i \ p(\overline{R}_i)$. This contradicts \eqref{prueba unanime 2}. Hence, $\psi$ is \emph{strategy-proof}. This proves the claim.
    \medskip

    \noindent Therefore, $\psi \in \Phi^\star$. Clearly, $\psi \succcurlyeq \varphi$ and $\varphi \not \succcurlyeq \psi$, contradicting that $\varphi \in \Phi^\star$. We conclude then that $\varphi$ is \emph{unanimous}.   
\end{proof}

\vspace{15 pt}

\noindent
\emph{Proof of Theorem \ref{caractvariosbienes}}.
The multidimensional uniform rule is  \emph{strongly Pareto-undominated strategy-proof} by  Corollary  \ref{Ufeenm} and, therefore, \emph{Pareto-undominated strategy-proof.} Moreover, it satisfies   \emph{equal treatment.} We already mentioned that it is    \emph{multidimensional replacement monotonic.}  Let  $\varphi$ be a rule that satisfies the properties listed in Theorem \ref{caractvariosbienes}. By Lemma \ref{lemma unanimity}, $\varphi$ is \emph{unanimous}. Then, in particular,  $\varphi$ satisfies \emph{strategy-proofness,} \emph{unanimity,} \emph{equal treatment} and \emph{non-bossiness}. It follows, from Corollary 1 in  \cite{Morimoto13},  that  $\varphi=u.$ \hfill $\square $

\end{document}